\documentclass[12pt]{article}
\usepackage[english]{babel}
\usepackage{amsmath,amsthm,amssymb}
\usepackage[utf8]{inputenc}
\usepackage{tikz,xinttools}
\usepackage{float, caption}
\usepackage{natbib}
\usetikzlibrary{calc}
\usepackage{graphicx}

\tikzstyle{every node}=[circle, draw, fill=black!70,
                        inner sep=0pt, minimum width=4pt]
\usepackage{listings}

\newcounter{nalg} 
\renewcommand{\thenalg}{\thechapter .\arabic{nalg}} 
\DeclareCaptionLabelFormat{algocaption}{Algorithm \thenalg} 

\lstnewenvironment{algorithm}[1][] 
{   
    \refstepcounter{nalg} 
    \captionsetup{labelformat=algocaption,labelsep=colon} 
    \lstset{ 
        frame=tB,
        numbers=left, 
        numberstyle=\tiny,
        basicstyle=\scriptsize, 
        keywordstyle=\color{black}\bfseries\em,
        keywords={,input, output, return, datatype, function, for, in, if, else, foreach, while, begin, end, } 
        numbers=left,
        xleftmargin=.04\textwidth,
        #1 
    }
}
{}

\newtheorem{lemma}{Lemma}
\newtheorem{theorem}{Theorem}

\theoremstyle{definition}
\newtheorem{definition}{Definition}

\theoremstyle{definition}

\title{Nonparametric geometric outlier detection}
\author{Matias Heikkilä\footnote{Aalto University School of Science}}

\begin{document}
\maketitle

\begin{abstract}
Outlier detection is a major topic in robust statistics due to the high practical significance of anomalous observations. Many existing methods are, however, either parametric or cease to perform well when the data is far from linearly structured. In this paper, we propose a quantity, Delaunay outlyingness, that is a nonparametric outlyingness score applicable to data with complicated structure. The approach is based a well known triangulation of the sample, which seems to reflect the sparsity of the pointset to different directions in a useful way. In addition to appealing to heuristics, we derive results on the asymptotic behaviour of Delaunay outlyingness in the case of a sufficiently simple set of observations. Simulations and an application to financial data are also discussed.
\end{abstract}

\emph{Running headline:} Nonparametric outlier detection


\section{Introduction}

Heuristically, an outlier is a point that is different from the majority of the data. According to \citet{hawkins}: "An outlier is an observation that deviates so much from other observations as to arouse suspicion that it was generated by a different mechanism." Due to their significance in various contexts, several methods for outlier detection have been proposed (\citet{outlierSurvey}).

In this paper, we propose a simple nonparametric approach to outlier detection that is based on the geometry of the sample. The method seems to successfully distinguish outliers from samples that are problematic for many existing outlier detection methods (see Figure \ref{fig::example}). We also derive a result, according to which, the method will asymptotically detect a finite set of points outside a sufficiently simple set.

The method is based on the Delaunay triangulation of a sample. In addition to having desirable theoretical properties and demonstrating good performance in both simulated and real-data applications, the method appeals to intuition. Implementation is also feasible, as computational aspects of the Delaunay triangulation have been extensively studied and efficient algorithms are known \citep{algo}. There is, however, rather little previous literature on Delaunay triangulations, or the dual structure, Voronoi diagrams, in relation to outliers (see e.g. \citet{previous}, \citet{previous2} and \citet{VoronoiPrev}).

The paper is organized as follows: Delaunay outlyingness is introduced in Section \ref{sec::outlyingness}. Certain asymptotic properties of Delaunay outlyingness in the case of a random variable taking values from a compact strictly convex set are considered in Section \ref{sec::theory}. We apply Delaunay outlyingness to simulated and real in Section \ref{sec::simus}. Concluding remarks are given in \ref{sec::discussion}. Technical proofs are given in the Appendix.

\begin{figure}[h]
\begin{minipage}{0.5\textwidth}

\includegraphics[scale=0.4]{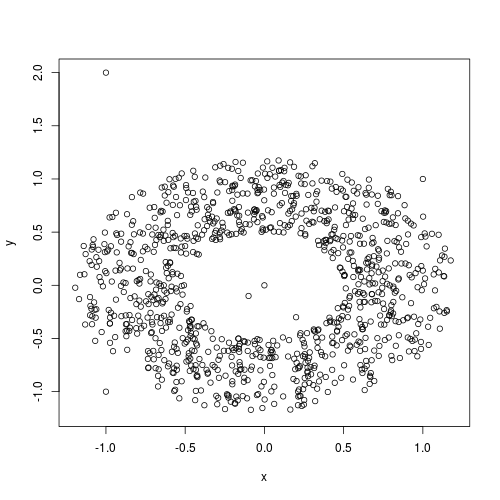}
\end{minipage}
\begin{minipage}{0.45\textwidth}

\includegraphics[scale=0.4]{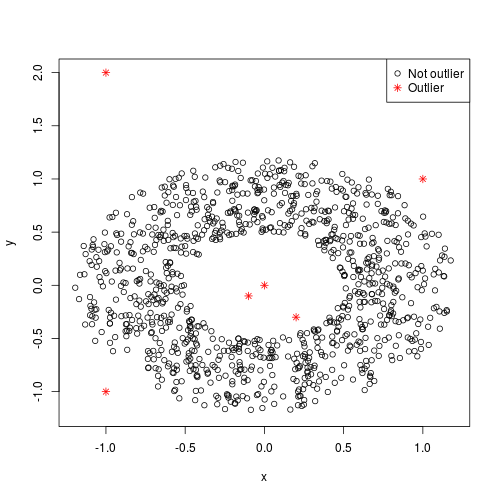}
\end{minipage}

\caption{Two plots of a sample. In the right plot the six points with the largest values of Delaunay outlyingness are highlighted.}
\label{fig::example}
\end{figure}

\section{Delaunay outlyingness}\label{sec::outlyingness}

The Delaunay triangulation is a well known object in geometry due to its special properties and relation to other interesting objects such as the minimum spanning tree \citep{geometers}. Well studied applications of Delaunay triangulations include density estimation \citep{phDThesis} and finite element methods \citep{meshGen}. 

We begin by reviewing the definition of the Delaunay triangulation. Degeneracies can be excluded as follows.

\begin{definition}\label{def::generalPosition}
A finite set of points $F \subset \mathbb{R}^k$ is said to be in general position if the points do not lie on a $k-1$ dimensional plane and no $k+2$ points are cospherical.
\end{definition}

\noindent
Let $X: \Omega \to \mathbb{R}^k$ be a random variable with a density. Then the set of values of $n > k$ i.i.d. copies of $X$ are in general position with probability one.

\begin{definition}
Let $F \subset \mathbb{R}^k$ be a finite set of points in general position. The Voronoi diagram of $F$ is the collection of sets $\left\{ V_x \right\}_{x \in F}$ defined by
\begin{equation*}
V_x = \left\{ y \in \mathbb{R}^k \: \middle| \: d(x,y) \leq d(z,y) \text{ for all }z \in F \right\},
\end{equation*}
where $d$ is the Euclidean metric. The sets $V_x$ are called Voronoi cells.
\end{definition}

For points in general position, the Delaunay triangulation can be uniquely defined as the dual of the Voronoi diagram (for reference see e.g. \cite{VoronoiDual} Chapter 23). Recall that a (finite) undirected graph is a pair $G=(F,E)$ with $F$ a finite set and $E \subset \left\{ \left\{ x,y \right\} \: \middle| \: x,y \in F \right\}$. The set $F$ is called the set of nodes of $G$ and the set $E$ is called the set of edges of $G$. Points $x,y \in F$ are said to be adjacent if $\left\{ x,y \right\} \in E$. 
\begin{definition}
Let $F \subset \mathbb{R}^k$ be a finite set of points in a general position and let $\left\{ V_x \right\}_{x \in F}$ be the Voronoi diagram of $F$. The Delaunay triangulation of $F$, denoted $\mathcal{DT}(F)$, is the undirected graph $(F,E)$ with the following property. If $x,y \in F$, $x \neq y$, then $\left\{ x,y \right\} \in E$ if $V_x \cap V_y \neq \varnothing$.
\end{definition}
\noindent
The duality between the Voronoi diagram and the Delaunay triangulation of $F$ is illustrated in Figure \ref{fig::VoronoiDelaunay}. Note that that points are adjacent if and only if the corresponding Voronoi cells $V_x$ share boundaries. Note also that the edges of the Delaunay triangulation tend to be shorter for points whose corresponding Voronoi cells are small. The heuristic is that the distance between a point and the edge of the corresponding Voronoi cell \emph{to a direction} is related to the nearest neighbour of the point to that direction (more specific statement given in Lemma \ref{lemma::checkAdj}).

\begin{figure}
\centering
\includegraphics[scale=0.5]{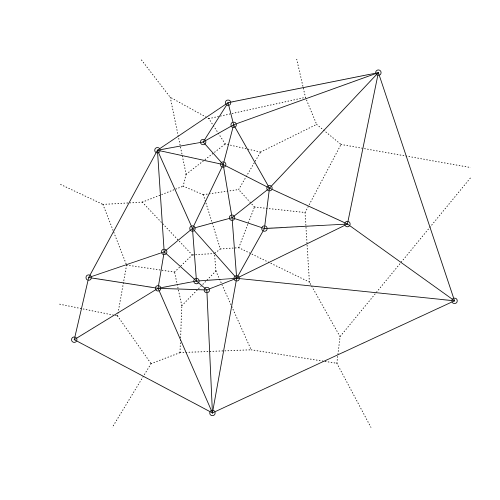}
\caption{A plot of a set of points $x$, the corresponding Voronoi diagram (boundaries of $V_x$ dotted lines) and the Delaunay triangulation of the points (solid lines).}
\label{fig::VoronoiDelaunay}
\end{figure}

We define some notation for convenience.
\begin{definition}\label{def::notationL}
Let $x,y \in \mathbb{R}^k$. We denote by $L(\left\{ x,y \right\})$ the line segment connecting the points $x$ and $y$, i.e.
\begin{equation*}
L(\left\{ x,y \right\}) = \left\{ \lambda x + (1-\lambda)y \: \middle| \: 0 \leq \lambda \leq 1 \right\}.
\end{equation*}
\end{definition}

\begin{definition}
Let $T \subset \mathbb{R}^k$ be a line segment, i.e. there are $x,y \in \mathbb{R}^k$ such that $T = L( \left\{x,y \right\})$. We denote the length of the line segment $T$ by $\lambda(T)$, i.e.
\begin{equation*}
\lambda(T) = d(x,y).
\end{equation*}
\end{definition}
\noindent
We propose the following quantity for an outlyingness measure of a point.
\begin{definition}\label{def::outlyingness}
Let $F \subset \mathbb{R}^k$ be a finite set of points in general position. Let $\mathcal{DT}(F) = (F,E)$. For $x \in F$, denote by $E(x)$ the set of edges $e \in E$ such that $x \in e$. The Delaunay outlyingness function $f_F: F \to \mathbb{R}$ is
\begin{equation*}
f_F(x) = \left( \prod \limits _{e \in E(x)} \lambda ( L(e)) \right)^{1 / \left| E(x) \right|}.
\end{equation*}
That is, the geometric mean of the lengths of the edges in $E(x)$.
\end{definition}

To illustrate the manner in which Delaunay outlyingness reflects the structure of the sample, consider the sample $F$ plotted in Figure \ref{fig::example}. Three plots of the values of $f_F$ are given in Figure \ref{fig::plotsOfExample}. In the leftmost plot all values of $f_F$ are displayed, in the middle plot only the values of $f_F$ at most $0.1$ are displayed and in the rightmost plot, only the values of $f_F$ at most $0.01$ are displayed.
\begin{figure}[h]
\begin{minipage}{0.30\textwidth}

\includegraphics[scale=0.25]{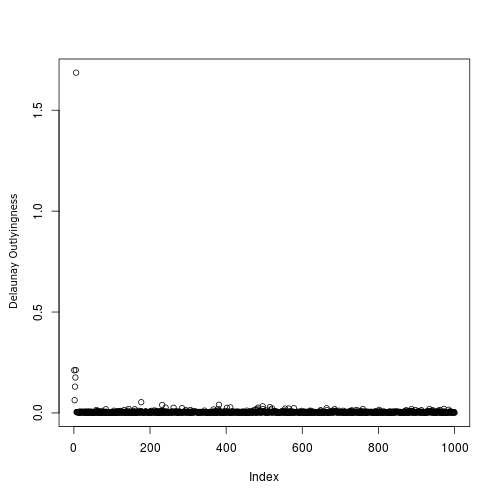}
\end{minipage}
\begin{minipage}{0.30\textwidth}

\includegraphics[scale=0.25]{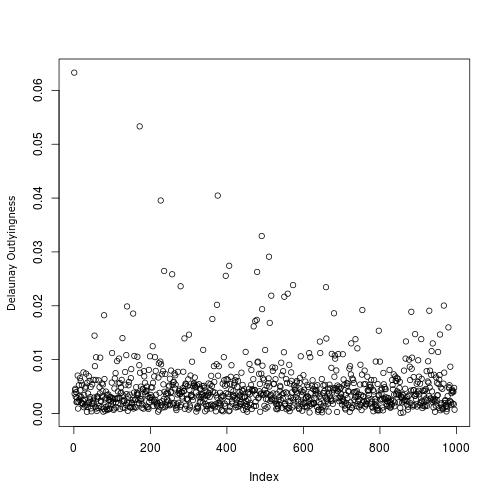}
\end{minipage}
\begin{minipage}{0.30\textwidth}

\includegraphics[scale=0.25]{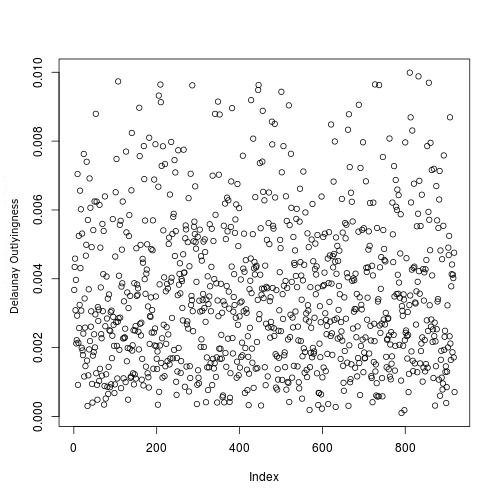}
\end{minipage}

\caption{Three plots of the values of the Delaunay outlyingness of the points displayed in Figure \ref{fig::example}.}
\label{fig::plotsOfExample}
\end{figure}

\section{A consistency property}\label{sec::theory}

\begin{definition}
A compact convex set $K \subset \mathbb{R}^k$ is a compact strictly convex set, if, for any $x,y \in K$ and $0 < \lambda < 1$, the point $\lambda x + (1- \lambda)y$ is an interior point of $K$.
\end{definition}

We establish a consistency property in the case of a compact strictly convex set: Consider a sample $F \cup U_n$ consisting of a finite set of points $F$ disjoint from a compact convex set $K$ and a set of $n$ observations $U_n \subset K$. Now for all $x \in F$, $f_{F \cup U_n}(x) > \varepsilon$ for some $\varepsilon > 0$ and all $n$. Simultaneously $f_{F \cup U_n}(x) \to_P 0$ for all $x \in U_n$.

We define some notation used throughout this section. Let $x \in \mathbb{R}^k$ and $S, T \subset \mathbb{R}^k$. Then
\begin{equation*}
d(x, S) = \inf \limits_{s \in S} d(x,s)
\end{equation*}
and
\begin{equation*}
d(S,T) = \inf \limits_{s \in S, t \in T} d(s,t).
\end{equation*}
We will also denote the boundary of a set $A \subset \mathbb{R}^k$ by $\partial A $.

\begin{definition}
Let $A \subset \mathbb{R}^k$. The boundary of $A$, denoted $\partial A$, is the set of points $x \in \mathbb{R}^k$ with the following property: Let $U \ni x$ be any open neighbourhood of $x$. Then $U \cap A \neq \varnothing$ and $U \cap \left(  \mathbb{R}^k \setminus A \right) \neq \varnothing$.
\end{definition}

\noindent
The boundary of any set is closed. Also, a compact set in $\mathbb{R}^k$ contains its boundary. The only subsets of $\mathbb{R}^k$ that have an empty boundary are $\varnothing$ and $\mathbb{R}^k$.

\begin{lemma}\label{lemma::strictlyConvex}
Let $K \subset \mathbb{R}^k$ be a compact strictly convex set and let $\varepsilon > 0$, $\varepsilon < \mathrm{Diam}(K)$. There is $\theta( \varepsilon )>0$, such that, for all $x, y \in K$, $d(x,y) = \varepsilon$,
\begin{equation}
d \left( \frac{1}{2}x + \frac{1}{2} y, \partial K \right) \geq \theta ( \varepsilon ).
\end{equation}
\end{lemma}

Note that $\theta(\varepsilon)$ is a monotone increasing function with respect to $\varepsilon> 0$. Lemma \ref{lemma::strictlyConvex} allows one to derive the following result.

\begin{lemma}\label{lemma::insideBound}
Let $K \subset \mathbb{R}^k$ be a compact strictly convex set. Let $X_n: \Omega \to K$ be a sequence of i.i.d random variables with a nonzero density over $K$ and denote the sample $\left\{ X_1(\omega), \dots, X_n ( \omega )\right\}$ by $U_n$. Consider the random Delaunay triangulations $\mathcal{DT}(U_n) =(U_n, E_n)$ and let
\begin{equation*}
\Lambda_n = \max \limits_{e \in E_n} \lambda( L(e)).
\end{equation*}
Then
\begin{equation*}
\Lambda_n \to_P 0.
\end{equation*}
\end{lemma}

\begin{theorem}\label{thm::consistency}
Let $K \subset \mathbb{R}^k$ be a compact strictly convex set. Let $X_n: \Omega \to K$ be a sequence of i.i.d random variables with a nonzero density over $K$ and denote the sample $\left\{ X_1(\omega), \dots, X_n ( \omega )\right\}$ by $U_n$. Let $F \subset \mathbb{R}^k$ be a finite set such that $d(F,K) > 0$. Consider the Delaunay outlyingness $f_{U_n \cup F}$. There is $\delta > 0$ such that
\begin{equation*}
f_{U_n \cup F}(x) > \delta
\end{equation*}
for all $x \in F$. Also, for all $u \in U_n$,
\begin{equation*}
f_{U_n \cup F} (x) \leq \gamma_n
\end{equation*}
with $\gamma_n \to_P 0$.
\end{theorem}

\section{Examples}\label{sec::simus}

In Section \ref{sec::pointSphere}, we consider a simple simulated setting, where we put a single point at the origin and draw points from an uniform distribution around the boundary of the unit sphere. In Section \ref{sec::prices} we apply the method to stock price data.

\subsection{A point at the center of the unit sphere}\label{sec::pointSphere}

Consider a spherically distributed random variable $X: \Omega \to \mathbb{R}^k$
\begin{equation*}
X = R \, \Theta,
\end{equation*}
where $R$ is uniformly distributed over the interval $[0.7, 1.1]$ and $\Theta$ is uniformly distributed on the unit sphere $\mathbb{S}^{k-1}$. The random variables $R$ and $\Theta$ are independent. We simulate the values of Delaunay outlyingness, when the sample $U_n$ consists of values of $n$ independent copies of $X$ and the set of outliers $F$ is the set $\left\{ 0 \right\}$. Almost similar setting in dimension 2 is displayed in Figure \ref{fig::example}.

In the first setting, we simulated 5000 samples $U_{299}$ of $X$ in dimension 4, each consisting of 299 observations. We then calculated the values of the Delaunay outlyingness $f_{U_{299} \cup F }(x)$ for all $x \in U_{299} \cup F$. In Figure \ref{fig::rel1}, we display a histogram of the corresponding ratios $f_{U_{299} \cup F}(x)/f_{U_{299} \cup F}(0)$, with $x \in U_{299}$. We call the ratio $f_{U_{299} \cup F}(x) / f_{U_{299} \cup F}(0)$ the relative outlyingness of $x \in U_{299}$.

Note that the vast majority of the observed relative outlyingness values are well below 1, i.e. the point at the origin is clearly separated from most of the points observed. During the 5000 rounds we observed 84 relative outlyingness values that were over 1 and 297 values that were over 0.9. The latter corresponds to $0.02 \%$ of the observed values.

\begin{figure}
\centering
\includegraphics[scale=0.45]{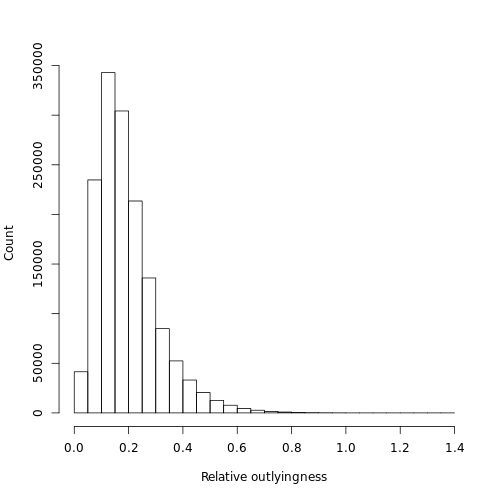}
\caption{A histogram of the obtained relative outlyingness values $f_{U_{299} \cup F}/f(0)$, with $x \in U_n$, obtained in dimension 4.}
\label{fig::rel1}
\end{figure}

We then compared the performance of our implementation in dimensions 3 and 5. We simulated 5000 samples of 199 observations and again calculated the relative outlyingness values. A comparison of the observed values is displayed in Figure \ref{fig::comparison}. In 3 dimensions we observed no relative outlyingness values over 1 and 5 values that were at least 0.9. The latter corresponds to less than one per million of the observed values. In 5 dimensions we observed 769 values over 1 and 2177 values that were at least 0.9.  The latter corresponds to $0.2 \%$ of the observed values. Thus the outlier at the origin is again distinguished with a rather high probability.

Note the high leftmost bar in the right histogram of Figure \ref{fig::comparison}. This seems to be due to small outlyingness scores rounding to zero in floating point arithmetic.

\begin{figure}[h!]
\begin{minipage}{0.4\textwidth}
\includegraphics[scale=0.33]{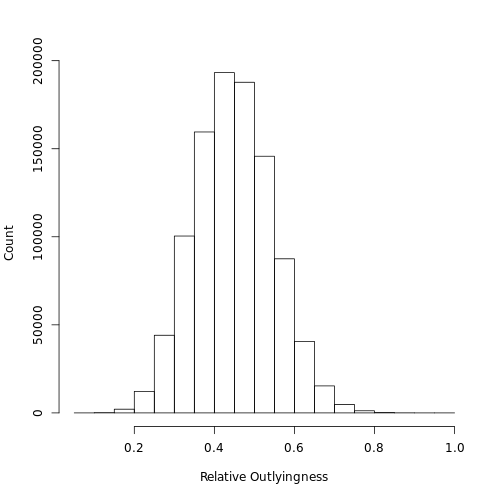}
\end{minipage}
\hspace{2cm}
\begin{minipage}{0.4\textwidth}
\includegraphics[scale=0.33]{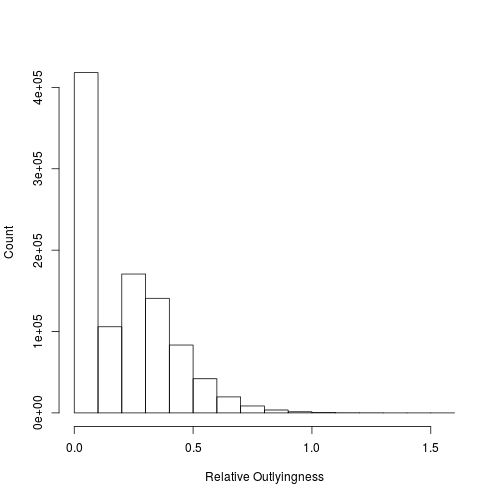}
\end{minipage}
\caption{A comparison of the relative outlyingness values $f_{U_{199}}/ f(0)$, with $x \in U_{199}$, obtained in dimension 3 (left) and 5 (right).}
\label{fig::comparison}
\end{figure}

\subsection{Stock prices}\label{sec::prices}

Let $u_i$ be the opening prices of General Electric (GE) Company stock and $v_i$ be the opening prices of Nokia (NOK) Corporation stock, where $i$ ranges from January 3 2007 to May 5 2016. A plot of points $(u_i, v_i)$ is displayed in Figure \ref{fig::geNok}. 

Let $T$ be the set of values of Delaunay outlyingness of the points $(u_i,v_i)$. Three plots of $T$ are displayed in Figure \ref{fig::delaunayGeNok}. In the leftmost one, all elements of $T$ are plotted. In the middle one, only the elements of $T$ that are at most $2$ are plotted. In the rightmost one, only the values of $T$ that are at most $0.2$ are plotted.

We identify outliers by setting a threshold $\alpha$ on the Delaunay outlyingness and flagging all points with Delaunay outlyingness at least $\alpha$ as outliers. In Figure \ref{fig::outliersGeNok}, we display the results yielded by three different choices of $\alpha$.

\begin{figure}
\includegraphics[scale=0.45]{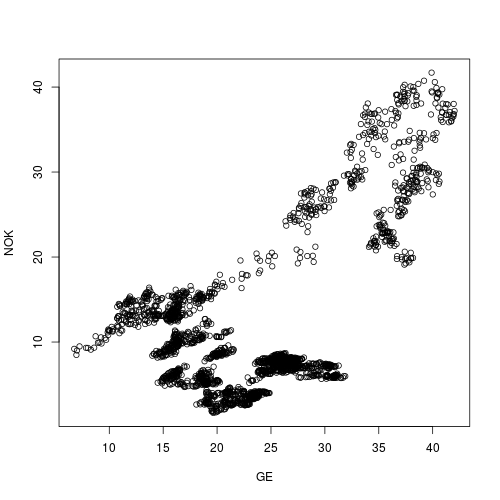}
\centering
\caption{The opening prices of NOK plotted against the opening prices of GE from January 3 2007 to May 5 2016.}
\label{fig::geNok}
\end{figure}

\begin{figure}[h!]
\begin{minipage}{0.30\textwidth}

\includegraphics[scale=0.25]{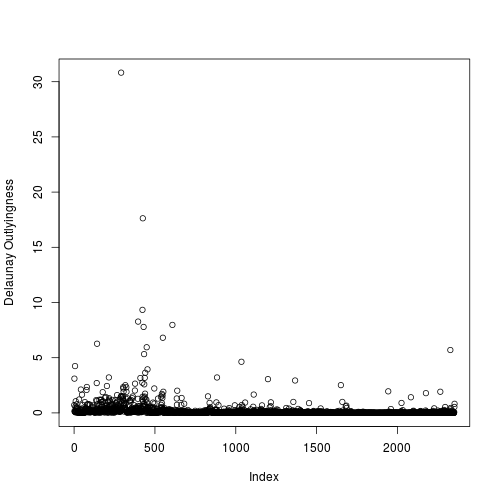}
\end{minipage}
\begin{minipage}{0.30\textwidth}

\includegraphics[scale=0.25]{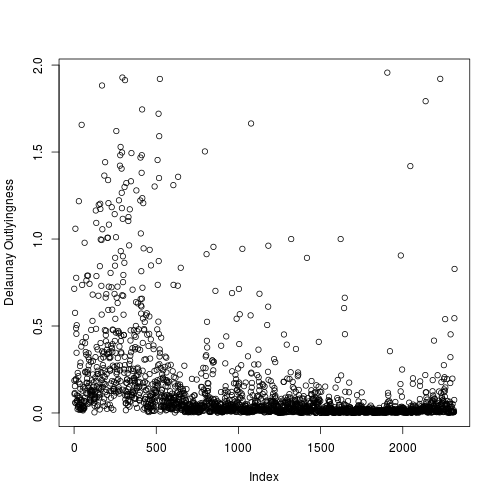}
\end{minipage}
\begin{minipage}{0.30\textwidth}

\includegraphics[scale=0.25]{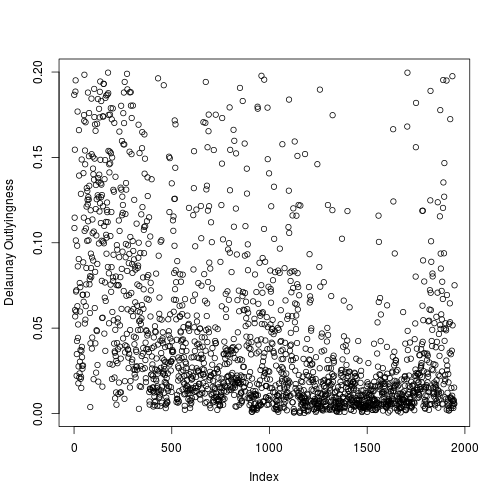}
\end{minipage}

\caption{Three plots of the values of the Delaunay outlyingness of the points displayed in Figure \ref{fig::geNok}.}
\label{fig::delaunayGeNok}
\end{figure}

\begin{figure}[h!]
\begin{minipage}{0.30\textwidth}

\includegraphics[scale=0.25]{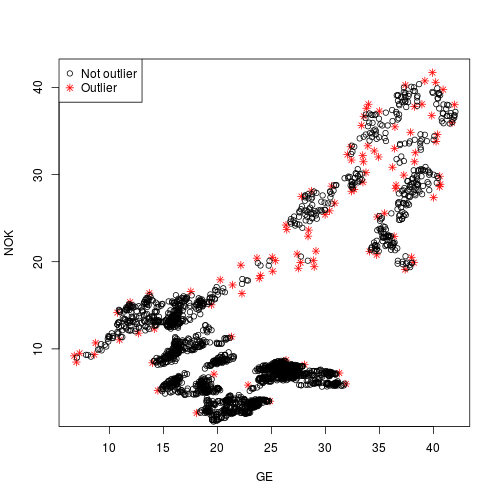}
\end{minipage}
\begin{minipage}{0.30\textwidth}

\includegraphics[scale=0.25]{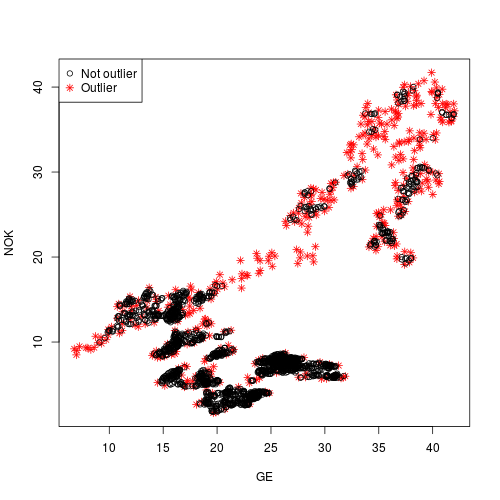}
\end{minipage}
\begin{minipage}{0.30\textwidth}

\includegraphics[scale=0.25]{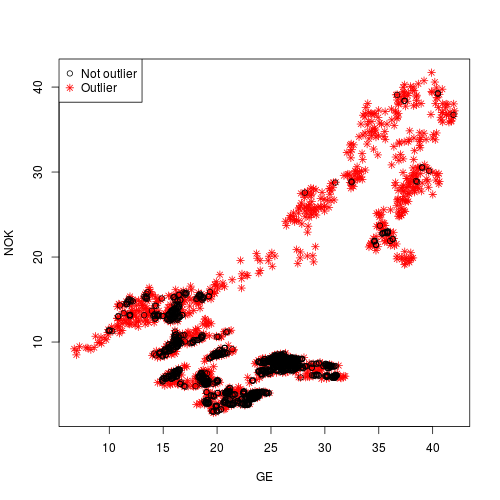}
\end{minipage}

\caption{Plot of the points appearing in Figure \ref{fig::geNok}, with points whose Delaunay outlyingness is at least $\alpha$ highlighted. In the leftmost plot $\alpha =1$, in the middle plot $\alpha=0.2$ and in the rightmost plot $\alpha=0.05$.}
\label{fig::outliersGeNok}
\end{figure}

\section{Discussion}\label{sec::discussion}

In this paper, we have proposed a quantity, Delaunay outlyingness, that can be used as an outlyingness score of observations. The quantity is almost surely uniquely defined for continuously distributed data and is thus provides a nonparametric approach to outlier detection. The behaviour of the quantity is simple enough to be analysed under a model, it is implementable and it appeals to intuition. A consistency result was given in a simple model of an outlier detection situation and we presented example applications in both simulated and real data.

Due to the unique and useful structure it endows a finite set with, the Delaunay triangulation could also find wider applications in statistical methodology as a data driven alternative to parametric approaches e.g. in change point analysis of multivariate time series (see Figures \ref{fig::timeSeries} and \ref{fig::timeSeriesDel}). Consider for example the classical $k$ nearest neighbour classifier in machine learning. Delaunay triangulation provides a possibility to omit the parameter $k$ and use the data given by the edges to a point in classification.

\begin{figure}[h!]
\begin{minipage}{0.30\textwidth}
\includegraphics[scale=0.25]{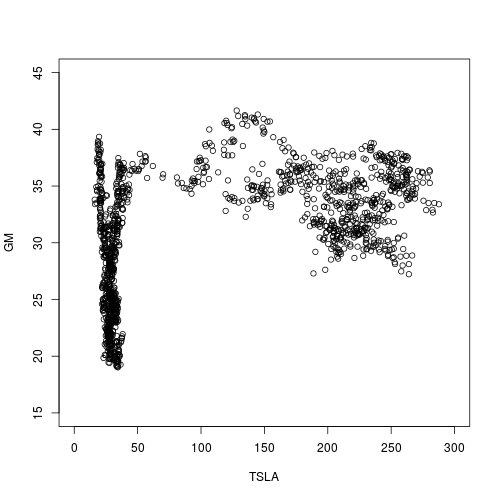}
\end{minipage}
\begin{minipage}{0.30\textwidth}
\includegraphics[scale=0.25]{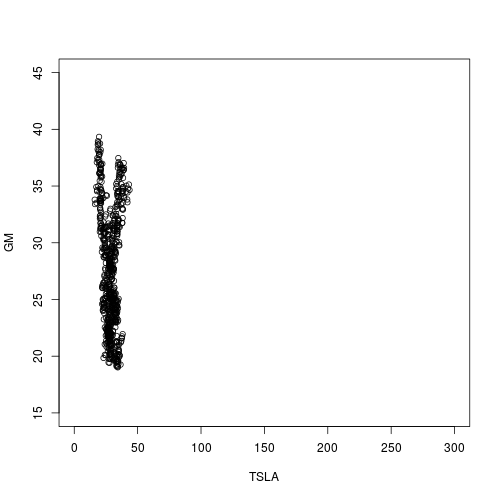}
\end{minipage}
\begin{minipage}{0.30\textwidth}
\includegraphics[scale=0.25]{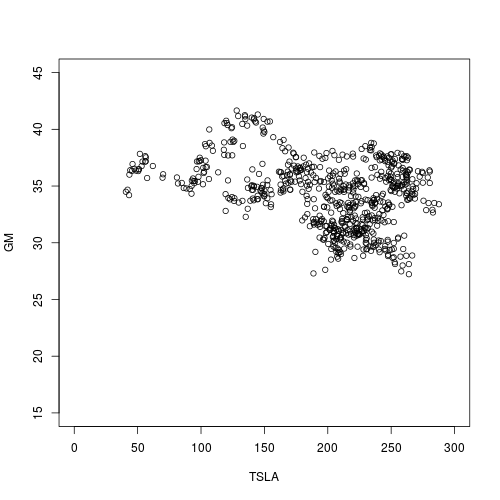}
\end{minipage}
\caption{Leftmost plot displays the opening prices of General Motors (GM) Company plotted against the opening prices of Tesla Motors (TSLA) Inc. from June 29 2010 to December 16 2015. In the middle plot only the prices of the first 700 days are displayed and in the rightmost plot only the prices of the days from 701st one onwards are displayed. Note that there is a clear change point near the 700th day and that the change point is visibile from Delaunay outlyingness values displayed in Figure \ref{fig::timeSeriesDel}.}
\label{fig::timeSeries}
\end{figure}

\begin{figure}[h!]
\centering
\includegraphics[scale=0.45]{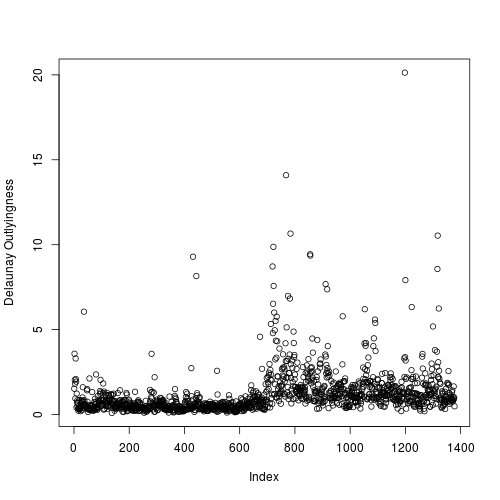}
\caption{The Delaunay outlyingness of the points displayed in the leftmost plot of Figure \ref{fig::timeSeries} with the index being the ordinal of the day in the observed interval. Note that there is a clear change point near the index 700.}
\label{fig::timeSeriesDel}
\end{figure}

\section*{Acknowledgements}

The author gratefully acknowledges support from the Magnus Ehrnrooth Foundation.

\bibliographystyle{humannat}
\bibliography{mybib}

\newpage

ADDRESS:

Matias Heikkilä

matias.heikkila@aalto.fi

Aalto University School of Science

Department of Mathematics and Systems Analysis

Otakaari 1

Espoo

Finland

\section*{Appendix: Proofs}

Note the following simple criterion for checking whether points are adjacent in the Delaunay triangulation.
\begin{lemma}\label{lemma::checkAdj}
Let $F \subset \mathbb{R}^k$ be a finite set of points in a general position and let $\mathcal{DT}(F) = (F,E)$. Let $x,y \in F$, $x \neq y$. Then $\left\{ x,y \right\} \in E$ if and only if there is a point $p$, equidistant from the points $x$ and $y$ and such that $d(p,z) \geq d(p,x)$ for all $z \in F$.
\end{lemma}

\begin{proof}
If $\left\{ x,y \right\} \in E$, $V_x \cap V_y \neq \varnothing$. Let $p \in V_x \cap V_y$. Now $p$ is, by definition, equidistant from $x$ and $y$. If a point $z \in V$ would be a strict lower bound for $\left\| x - p \right\|_2$, we would have $x \in V_z \setminus V_x$ which is a contradiction.

Let $p \in V_x \cap V_y$. If a point $p$ as above exists, $p \in V_x \cap V_y$ and $V_x \cap V_y \neq \varnothing$.
\end{proof}

\noindent
In particular, this implies that if $y$ is the nearest neighbour of $x$, then $x$ and $y$ are adjacent. Note also the following property of the possible test points $p$ in Lemma \ref{lemma::checkAdj}.

\begin{lemma}\label{lemma::plane}
Let $x,y \in \mathbb{R}^k$. Let $p \in \mathbb{R}^k$ be equidistant from $x$ and $y$. Then
\begin{equation*}
d(x,p) = \sqrt{ d(x,a)^2 + d(a,p)^2},
\end{equation*}
with $a = \frac{1}{2}x+\frac{1}{2}y$. In particular,
\begin{equation*}
d(x,p) \geq d(x,a), d(a,p).
\end{equation*}
\end{lemma}

\begin{proof}
The set of points equidistant from $x$ and $y$ form a $k-1$ dimensional plane $T$. Note that $a \in T$. Now, note that $a$ is the closest point of $T$ to $x$ and $y$ and the vector $x-a$ is thus perpendicular to the plane $T$. Now the line from $a$ to $p$ is contained in the plane and $p-a \perp x-a$. Thus
\begin{eqnarray*}
d(x,p)^2 &=& \left\langle x-p , x-p \right\rangle \\
 &=& \left\langle (x-a)+(a-p) , (x-a)+(a-p) \right\rangle\\
 &=& \left\langle x-a , x-a \right\rangle + 2 \left\langle x-a, a-p\right\rangle + \left\langle a-p, a-p \right\rangle \\
 &=& d(x,a)^2 + d(a,p)^2.
\end{eqnarray*}
\end{proof}

\begin{proof}[Proof of Lemma \ref{lemma::strictlyConvex}]
The metric
\begin{equation}\label{eq::productMetric}
d_{K \times K}((x,y),(x',y')) = d(x,x')+d(y,y')
\end{equation}
induces the product topology on $K \times K$. It is well known that, as a product of compact topological spaces, $K \times K$ is compact with respect to the product topology.

Consider a function $s: K \times K \mapsto \mathbb{R}$ defined by
\begin{equation*}
s(x,y) = d \left( \frac{1}{2}x + \frac{1}{2} y, \partial K \right).
\end{equation*}
The function can be expressed as $s = u \circ t$ with $u: K \to \mathbb{R}$ defined as $u(x) = d(x,\partial K)$ and $t: K \times K \to K$ defined as $t(x,y) = \frac{1}{2}x+\frac{1}{2}y$. Now $t$ is continuous with respect to the metric \eqref{eq::productMetric} and $u$ is continuous in $K$. Thus $s$ is continuous as a composition of continuous functions.

Let $\varepsilon > 0$, $\varepsilon < \mathrm{Diam}(K)$, and
\begin{equation*}
U_\varepsilon = \left\{ (x,y) \in K \times K \: \middle| \: d(x,y) = \varepsilon \right\}.
\end{equation*}
The set $U_\varepsilon$ is nonempty, since $\varepsilon < \mathrm{Diam}(K)$. Consider a convergent sequence $(x_n,y_n) \to (x,y)$, whose each element is in $U_\varepsilon$. Since the metric $d$ is a continuous function in $K \times K$, $d(x,y)$ is the limit of a sequence of constants $d(x_n,y_n) = \varepsilon$ and, consequentially, $d(x,y) = \varepsilon$. Thus $(x,y) \in U_\varepsilon$ and $U_\varepsilon$ is closed for all $\varepsilon >0$.

Now $U_\varepsilon$ is compact as a closed subset of a compact space. Thus the continuous function $s$ attains its infimum over $U_\varepsilon$. Denote this number by $\theta(\varepsilon)$. Let $x,y \in K$ be such that
\begin{equation*}
s(x,y) = d \left( \frac{1}{2}x + \frac{1}{2} y, \partial K \right) = \theta (\varepsilon)
\end{equation*}
and assume that $\theta (\varepsilon) = 0$. There is a point $p \in \partial K$ that minimizes the distance of $\partial K$ to $\frac{1}{2} x + \frac{1}{2} y$. This is because distance to the point $\frac{1}{2}x +\frac{1}{2}y$ is a continuous function in $K$ and $\partial K$ is a compact set. Thus $d \left( \frac{1}{2}x + \frac{1}{2} y, p\right) = 0$ and $p = \frac{1}{2}x + \frac{1}{2}y$. This is a contradiction, since $p \in \partial K$ and the set $K$ is strictly convex. Thus $\theta ( \varepsilon ) > 0$ for all $\varepsilon > 0$.
\end{proof}

\begin{proof}[Proof of Lemma \ref{lemma::insideBound}]
Let $\varepsilon > 0$. Let $\delta = \min \left\{ \frac{\varepsilon}{8}, \frac{\theta(\varepsilon)}{4}  \right\}$, with $\theta(\varepsilon)$ as in Lemma \ref{lemma::strictlyConvex}. The set
\begin{equation*}
\mathcal{U} = \left\{ B(x, \delta) \: \middle| \: x \in K \right\}
\end{equation*}
is an open cover of a compact set $K$. Thus, there is a finite set $\mathcal{F}\subset \mathcal{U}$ such that $\cup \mathcal{F} \supset K$. Since $X_n$ has a nonzero density over $K$ and $K$ is convex, we can select the cover $\mathcal{F}$ so that $\mathbb{P}(X_n \in A) > 0$ for all $A \in \mathcal{F}$. Asymptotically almost surely, there is an observation in each $A \in \mathcal{F}$. I.e.
\begin{equation*}
U_n \cap A \neq \varnothing,
\end{equation*}
for all $A \in \mathcal{F}$. Assume that this is the case. We can now show that there  can not be a an edge $e \in E_n$ with length $\varepsilon$.

Let $x,y \in U_n$ with $d(x,y) \geq \varepsilon$. According to Lemma \ref{lemma::checkAdj}, if ${x,y} \in E_n$, there is a point $p \in \mathbb{R}^k$ such that $x$ and $y$ are equidistant from $p$ and they are also minimize the distance between an element of $U_n$ and $p$.

Assume that $p \in E$. Let $a = \frac{1}{2}x+\frac{1}{2}y$. By Lemma \ref{lemma::plane}
\begin{equation*}
d(x,p) \geq d(x,a) = \frac{\varepsilon}{2}.
\end{equation*}
However, $p$ is contained in a ball $A \in \mathcal{F}$ of radius at most $\frac{\varepsilon}{8}$ along with an observation $z \in U_n, z \neq x,y$. By triangle inequality
\begin{equation*}
d(z,p) \leq \frac{\varepsilon}{4}.
\end{equation*}
Thus $x$ and $y$ do not minimize the distance between an element of $U_n$ and $p$ if $p \in K$.

Assume that $p \not\in K$ and consider the line segment $L({a,p})$. There is $q \in L({a,p})$ such that $q \in \partial K$. Now by Lemma \ref{lemma::plane}
\begin{equation*}
d(x,p) \geq d(a,p) = d(a,q) + d(q,p) \geq \theta( \varepsilon) + d(q,p),
\end{equation*}
with $\theta(\varepsilon)$ as in Lemma \ref{lemma::strictlyConvex}. Since, in $\mathcal{F}$, there is a ball $A$ of radius at most $\frac{\theta(\varepsilon)}{4}$ with $A \cap U_n \neq \varnothing$, there is, by triangle inequality, a point $z \in U_n$ with
\begin{equation*}
d(z,p) \leq \frac{\theta(\varepsilon)}{2} + d(q,p) < d(x,p).
\end{equation*} 
Thus $x$ and $y$ are not the two nearest points of $U_n$ to $p \not\in K$.
\end{proof}

\begin{proof}[Proof of Theorem \ref{thm::consistency}]
Let $\mathcal{DT}(U_n\cup F) = (U_n \cup F, E_n)$. Denote, for $x \in U_n \cup F$, by $E_n(x)$ the set of edges $e \in E_n$ such that $x \in e$.

Let
\begin{equation*}
\delta = \min \left\{ d(F,K), \min \limits_{x,y \in F} d(x,y) \right\} > 0.
\end{equation*}Now, for $x \in F$,
\begin{equation*}
f_{F \cup U_n} (x) = \left( \prod \limits_{e \in E_n(x)} \lambda (L(e)) \right)^{1 / \left| E_n(x) \right|} \geq \left( \delta^{E_n(x)} \right)^{1 / \left| E_n(x) \right|} = \delta.
\end{equation*}

Asymptotically almost surely, the nearest neighbour of every $x \in U_n$ is also in $U_n$. This can be seen by covering $K$ with a finite set of cubes with side length $a = \frac{\delta}{4 \sqrt{d}}$. Since $K$ is connected, in such covering, a cube always shares an edge with another cube. When there is an observation in each cube,
\begin{equation}\label{eq::nearestNeighbour}
\min \limits _{y \in U_n \setminus \left\{ x \right\} }d (x,y)\leq \frac{\delta}{2},
\end{equation}
for all $x \in U_n$. Now, by Lemma \ref{lemma::checkAdj} and triangle inequality, a point is adjacent to its nearest neighbour in the Delaunay triangulation.

Assume that the asymptotically almost sure bound \ref{eq::nearestNeighbour} holds and consider a point $x \in U_n$. The nearest neighbour of $x$ in $U_n \cup F$ then is a point $y \in U_n$ and $e =\left\{x,y \right\} \in E_n$. Now, since $x,y \in K$,
\begin{equation*}
\lambda(L(e)) = d(x,y) \leq \Lambda_n,
\end{equation*}
where $\Lambda_n$ is as in Lemma \ref{lemma::insideBound}.

Let
\begin{equation*}
\Delta = \mathrm{Diam}(K \cup F)
\end{equation*}
and 
\begin{equation*}
E_n^*(x) = \left\{ {x,y} \in E_n(x) \: \middle| \: y \in F \right\}, \quad  E_n^\circ(x) = \left\{ {x,y} \in E_n(x) \: \middle| \: y \in U_n \right\}.
\end{equation*}
Denote $\left| E_n(x) \right| = i$ and $\left| E^*_n(x) \right| = j$, then $\left| E^\circ_n(x) \right| = i-j$. Since the nearest neighbour of $x$ is in $U_n$, $i-j \geq 1$. Note that $j \leq \left| F \right|$.
\begin{eqnarray*}
f_{U_n \cup F} (x) &=& \left( \prod \limits_{e \in E_n(x)} \lambda (L(e)) \right)^{1 / \left| E_n(x) \right|} \\ &=& \left( \prod \limits_{e \in E^*_n(x)} \lambda (L(e)) \prod \limits_{e \in E^\circ_n(x)} \lambda (L(e))\right)^{1 / \left| E_n(x) \right|} \\ &\leq& \left( \Delta^{j} \Lambda_n^{i-j} \right)^{1 / i} \\
&\leq & \max \left\{ 1 , \Delta^{\left| F \right|} \right\} \Lambda_n^{1-\frac{j}{i}}
\end{eqnarray*}
Assume that $\Lambda_n < 1$, which is an asymptotically almost sure event and note that $j \leq \left| F \right|$ and $i \geq j +1$. Then
\begin{equation*}
1 - \frac{j}{i} \geq \min \left\{ 1 - \frac{s}{1+s} \: \middle| \: s \in \left\{ 1 , \dots, \left| F \right| \right\} \right\} = \frac{1}{1+\left| F \right|}
\end{equation*}
and
\begin{equation*}
\max \left\{ 1 , \Delta^{\left| F \right|} \right\} \Lambda_n^{1-\frac{j}{i}} \leq \max \left\{ 1 , \Delta^{\left| F \right|} \right\}\Lambda_n^\frac{1}{1+\left| F \right|} = g_n \to_P 0.
\end{equation*}
We may now define $\gamma_n$. When $\Lambda_n \geq 1$ or the bound \eqref{eq::nearestNeighbour} is invalid let $\gamma_n$ be any sufficiently large number. When the conditions hold, let $\gamma_n =g_n$. Since the conditions for $\gamma_n = g_n$ are asymptotically almost surely valid, $\gamma_n \to_P 0$.
\end{proof}

\end{document}